\documentclass[conference]{IEEEtran}
\IEEEoverridecommandlockouts
\usepackage{cite}
\usepackage{amsmath,amssymb,amsfonts}
\usepackage{algorithmic}
\usepackage{graphicx}
\usepackage{textcomp}
\usepackage{xcolor}

\usepackage{amsthm}
\usepackage{comment}
\usepackage[nolist]{acronym}

\usepackage{pgfplots}
\pgfplotsset{compat=1.18}

\definecolor{epszero}{RGB}{0,0,255}      % Blue
\definecolor{epsfive}{RGB}{255,128,0}    % Orange
\definecolor{epstwenty}{RGB}{0,128,0}    % Green

\newtheorem{definition}{Definition}

\newtheorem{example}[definition]{Example}
\newtheorem{theorem}[definition]{Theorem}

\newtheorem{remark}[definition]{Remark}
\newtheorem{corollary}[definition]{Corollary}
\newtheorem{proposition}[definition]{Proposition}

\renewcommand{\thealgorithm}{\Alph{algorithm}}

\newcommand{\cond}{\mathsf{p}}
\newcommand{\bin}{\ensuremath{\operatorname{Bin}}}

\def\BibTeX{{\rm B\kern-.05em{\sc i\kern-.025em b}\kern-.08em
    T\kern-.1667em\lower.7ex\hbox{E}\kern-.125emX}}

\begin{document}

\title{Probabilistic GC-Constraints for Composite DNA
\vspace{-0.2cm}
%%% Many authors with many affiliations:
\author{%
  \IEEEauthorblockN{
    \textbf{Humeyra Bodur}\IEEEauthorrefmark{1},
    \textbf{Frederik Walter}\IEEEauthorrefmark{1}, and
    \textbf{Antonia Wachter-Zeh}\IEEEauthorrefmark{1}
                    }
  \IEEEauthorblockA{\IEEEauthorrefmark{1}%
                   School of Computation, Information and Technology,  Technical University of Munich, Germany,
                    }
    
    Emails: \{humeyra.bodur, frederik.walter, antonia.wachter-zeh\}@tum.de
    \vspace{-0.6cm}
}
\thanks{Funded by the European Union (DiDAX, 101115134). Views and opinions expressed are however those of the author(s) only and do not necessarily reflect those of the European Union or the European Research Council Executive Agency. Neither the European Union nor the granting authority can be held responsible for them.}
}

% Acronyms 
\begin{acronym}
    \acro{nbc}[NBC]{natural binary code}
    \acro{pcr}[PCR]{\emph{Polymerase Chain Reaction}}
    \acro{pmf}[PMF]{\emph{probability mass function}}
    \acro{rll}[RLL]{run-length-limited}
\end{acronym}

\maketitle

\begin{abstract}
This paper addresses the challenge of encoding biochemical GC-content constraints in composite DNA-based data storage. Previous deterministic models impose high rate penalties by avoiding any possibility for a strand to fall outside the allowed range. To account for the stochastic nature of composite DNA, we introduce an $\epsilon$-probabilistic constraint framework, and derive capacity bounds for global constraints using composition types and for local sliding-window constraints via finite-state Markov chains. Furthermore, we propose a capacity-achieving multi-type enumerative encoder. 
\end{abstract}

\begin{IEEEkeywords}
DNA-based data storage, composite DNA, constrained coding, GC-content, channel capacity
\end{IEEEkeywords}

\section{Introduction}

% \IEEEPARstart{}{} creates a tall first letter for this first paragraph
\IEEEPARstart{I}{n} today’s information-driven world, efficiently storing large amounts of data for long periods is a significant challenge. DNA-based data storage has recently emerged as a promising solution for this \cite{Church2012DNAStorage,goldman2013}. It takes advantage of the natural ability of DNA molecules to encode vast amounts of information in an extremely compact form.

The classical approach to storing information in DNA, without error correction, maps a binary alphabet $\{0,1\}$ to the four nucleotides $\{A,C,G,T\}$, thereby encoding information directly into DNA sequences. However, commonly used DNA synthesis does not produce only a single copy of the sequence. Instead, the synthesis process inherently generates multiple copies of the same strand \cite{sabary2024reconstruction}. 
Many DNA-based data storage schemes treat this multiplicity primarily as redundancy for error correction, which reduces the effective information density per synthesized position \cite{anavy2019,Church2012DNAStorage}.

Composite DNA has been proposed as an alternative encoding strategy to exploit the many-copy nature of DNA synthesis~\cite{anavy2019}. In composite DNA, each position in a DNA strand represents a probabilistic composition of nucleotides instead of encoding a single, fixed base. Specifically, a composite symbol at position $i$ can be represented as a tuple $(p_A, p_C, p_G, p_T)$, where each $p_x$ denotes the relative proportion of nucleotide $x \in \{A, C, G, T\}$ and $p_A + p_C + p_G + p_T = 1$. Figure~\ref{fig:classical_vs_composite} illustrates a side-by-side comparison of classical and composite DNA encoding schemes, highlighting the structural differences between the two approaches. 

\begin{figure}[t]
    \centering
    \includegraphics[width=0.7\linewidth]{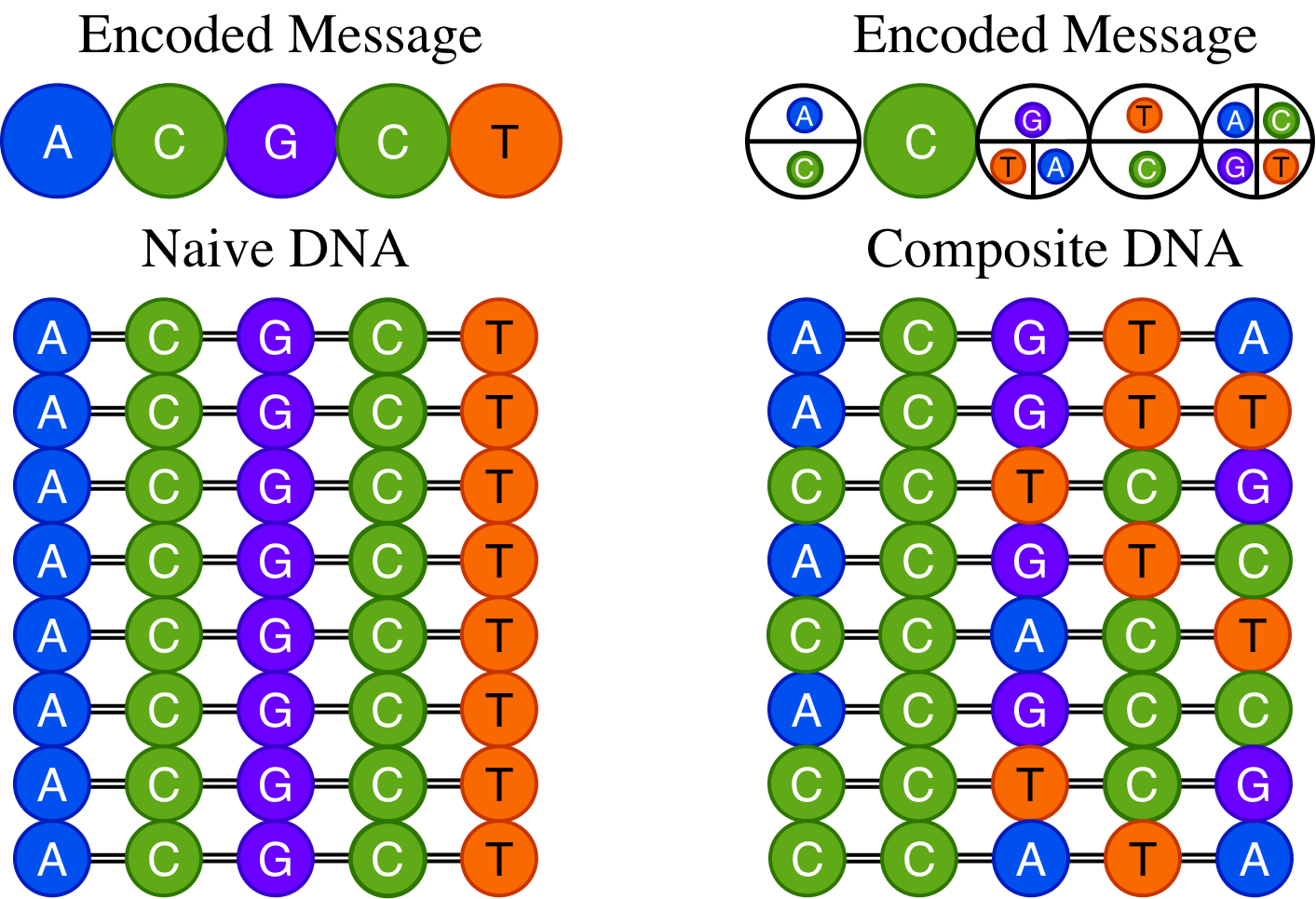}
    \caption{Classical (left) and composite (right) DNA encoding schemes.}
    \label{fig:classical_vs_composite}
\end{figure}

While composite DNA offers gains in theoretical information density, it also introduces fundamental challenges to ensure that synthesized sequences satisfy essential biochemical constraints. To secure reliable chemical processing, DNA sequences must avoid extreme variations in GC content. Excessively GC-rich or AT-rich sequences suffer from skewed melting temperatures, leading to amplification bias, non-uniform read coverage, and sequence dropout \cite{Erlich2017DNAFountain, Heckel2019Characterization, Yazdi2016DNAStorage}. 

As in traditional optical recording and communication systems, where physical channel properties restrict allowable sequences, constrained coding provides a general information-theoretic framework for accommodating these biochemical limits \cite{optical_constrained}. Rather than generating arbitrary sequences and blindly discarding unsuitable ones, this approach ensures that digital data is mapped directly and exclusively to valid, biologically compliant strands.

In classical DNA storage, constraint evaluation is deterministic. The transition to composite symbols introduces a different problem: the physical nucleotide sequence obtained during synthesis is probabilistic. Despite this, constrained coding for composite DNA remains relatively underexplored. Recently, Nguyen et al. \cite{nguyen} investigated both \ac{rll} and global GC-content constraints in composite DNA, calculating constrained capacities for small composite alphabets using adjacency matrices. Their capacity evaluations and subsequent code constructions relied on a strict, worst-case deterministic framework utilizing sequence replacement and Knuth's algorithm. A theoretical analysis of \ac{rll} coding was presented in \cite{walter2025coding, dang2026rll}.

A critical limitation of existing deterministic approaches is that enforcing zero-tolerance compliance against worst-case realizations introduces high rate penalties. In this paper, we shift to a \emph{probabilistic} constraint. This shift is justified by the physical and informational realities unique to the stochastic nature of composite DNA.

By allowing a controlled probability of constraint violation (denoted by $\epsilon$), we are able to expand the admissible code space. In this work, we develop a comprehensive framework for the analysis and encoding of composite DNA sequences subject to $\epsilon$-probabilistic GC constraints. We determine the capacity and
propose an efficient multi-type arithmetic indexing scheme that achieves it. The remainder of this paper details the constraint formulations (Section \ref{sec:preliminaries}), capacity formulations and analysis (Section \ref{sec:capacity}), and the proposed encoding scheme (Section \ref{sec:coding_schemes}). Proofs and examples are deferred to the appendix; short sketches are given for the main theorems.

\section{Model and Preliminaries}
\label{sec:preliminaries}

\subsection{Composite DNA Representation and Probability Bins}

The base alphabet for composite DNA consists of the four base nucleotides $\mathcal{X} = \{\text{A, C, G, T}\}$ \cite{anavy2019}. 

\begin{definition}[Composite DNA Symbol]
A composite DNA symbol $S$ over the alphabet $\mathcal{X}$ is defined by the integer counts
\[
S = (n_A, n_C, n_G, n_T), \qquad \sum_{x \in \mathcal{X}} n_x = k,
\]
where $k$ is called the resolution parameter and $n_x$ represents the count of nucleotide $x$ used to form the composite symbol. Equivalently, $S$ represents an empirical distribution $p_S(x) = n_x / k$ over $\mathcal{X}$. We denote the set of all composite symbols with resolution $k$ as the composite alphabet $\mathcal{A}_k$.
\end{definition}

The size of the composite alphabet is $|\mathcal{A}_k| = \binom{k+3}{3}$.
For the synthesis, a composite sequence $\mathbf{S} = [S_1, \ldots , S_n]$ is fed to the synthesizer. 
For each of the multiple synthesized strands, at position $i$, a nucleotide is chosen at random according to the distribution given by $S_i$.
We emphasize that the synthesis process does not provide control over which specific strand receives which nucleotide. Only the overall composition of nucleotides at each position is specified by the composite symbol. As a result, each nucleotide at position $i$ in a strand can be viewed as an independent realization drawn from the distribution associated with $S_i$.

Assume we sequence $M$ strands from the pool of a synthesized composite DNA sequence of length $n$. 
We represent the $m$-th strand as
\[
\mathbf{D}_m = [D_{m,1}, D_{m,2}, \dots, D_{m,n}] \in \mathcal{X}^n, \quad m \in \{1, \dots, M\},
\]
where $D_{m,i}$ denotes the nucleotide of strand $m$ at position $i$. In this view, the array of $M$ strands is the $M \times n$ matrix $\mathbf{D}$ whose $m$-th row is $\mathbf{D}_m$.
Note that each $D_{m,i}$ is a realization of a random variable over the distribution given by $S_i$.
To evaluate sequences efficiently, we group symbols by their intrinsic GC probability mass as follows. 

\begin{definition}[Probability-Mass Binning]
\label{def:binning}
For a resolution parameter $k$, the \emph{GC probability mass} is defined as $q(S) = p_S(G) + p_S(C)$. As it takes discrete values in $\{0, \frac{1}{k}, \dots, 1\}$,  we partition $\mathcal{A}_k$ into $k + 1$ disjoint probability bins $\mathcal{Q} = \{\mathcal{B}_0, \dots, \mathcal{B}_k\}$, where all symbols in bin $\mathcal{B}_i$ share the same GC probability mass, meaning
$\mathcal{B}_i = \{S \in \mathcal{A}_k : q(S) = \frac{i}{k} \}$.
\end{definition}

The cardinality of each bin is determined by the number of ways to independently partition $i$ elements into the $\{G, C\}$ bases and $k-i$ elements into the $\{A, T\}$ bases, yielding
\begin{equation}
|\mathcal{B}_i| = (i+1)(k-i+1).
\end{equation}

\subsection{GC-Constraint and the Probabilistic Shift}

Let \(\mathbb{I}_{\min}(S)\) and \(\mathbb{I}_{\max}(S)\) be indicator functions representing the guaranteed minimum and potential maximum physical GC contributions of a composite symbol \(S\):
\vspace{-0.2cm}
\begin{align}
    \mathbb{I}_{\min}(S) = 
    \begin{cases} 1, & \!\!\!\! \text{if } q(S) \! = \! 1 \\ 
    0, & \!\!\!\! \text{otherwise} \end{cases},
    \mathbb{I}_{\max}(S) = 
    \begin{cases} 1, & \!\!\!\! \text{if } q(S) \! > \! 0 \\ 
    0, & \!\!\!\! \text{otherwise.} \end{cases}
\end{align}
\vspace{-0.4cm} \\
Using these indicator functions, we define the deterministic GC constraint as follows: 

\begin{definition}[Deterministic GC Constraint]\label{def:det_gc}

    Let \(\mathbf{S} = [S_1, \dots, S_n] \in \mathcal{A}_k^n\). We say \(\mathbf{S}\) complies with the deterministic GC constraint with limits \(\alpha \le \beta \in [0,1]\) if
    \vspace{-0.2cm}
    \begin{align}
        n\alpha \leq \sum_{i=1}^{n} \mathbb{I}_{\min}(S_i) \quad \text{and} \quad \sum_{i=1}^{n} \mathbb{I}_{\max}(S_i) \leq n \beta.
    \end{align}
    \vspace{-0.4cm} \\
    We further say that \(\mathbf{S}\) complies with the \(w\)-windowed deterministic GC constraint with limits \(\alpha \le \beta \in [0,1]\) if any subsequence \(\mathbf{S}' = [S_i, \ldots , S_{i+w-1}]\) for \(i \in \{ 1, \ldots , n-w+1 \} \) complies with the deterministic GC constraint.
\end{definition}
While strict deterministic constraints are necessary for classical DNA storage, applying this worst-case, zero-tolerance approach to composite DNA causes rate loss. This strictness ignores two practical realities of composite synthesis.

(i) Random Chemistry \& Redundancy: DNA synthesis is a random process that creates millions of copies of each sequence. A worst-case physical strand is a rare anomaly, not a guaranteed outcome. (ii) Isolated Failures: If a rare, highly biased strand does form and drop out, it is an isolated erasure. It does not corrupt the rest of the molecular pool and can potentially be recovered using suitable error correction codes.
The following example illustrates why shifting to a probabilistic threshold is much more practical.

\begin{example}[Pessimism of Deterministic Bounds]
    Let \(k=4\), and let the standard acceptable GC limits for a sequence of length \(n=10\) be \(\alpha = 0.4\) and \(\beta = 0.6\). This requires the physical strand to contain between 4 and 6 GC nucleotides. 
    Consider a composite sequence \(\mathbf{S}\) consisting of 4 pure G/C symbols (\(q(S)=1\)), 3 pure A/T symbols (\(q(S)=0\)), and 3 mixed symbols biased toward A/T (e.g., 3 adenines and 1 cytosine, yielding \(q=0.25\)).
    
    Under the deterministic framework, we evaluate the worst-case bounds. The 4 pure G/C symbols guarantee the lower bound (\(\sum \mathbb{I}_{\min} = 4 \ge 4\)). However, the deterministic upper bound assumes all 3 mixed symbols will, in the worst case, manifest as cytosine. This yields a maximum potential weight of \(\sum \mathbb{I}_{\max} = 4 + 3 = 7\). Since \(7 > 6\), the sequence \(\mathbf{S}\) is rejected as invalid.
    
    Analyzing the strand probabilistically reveals a different reality. As the pure symbols are fixed, the sequence only violates the upper bound if \textit{all three} mixed symbols simultaneously manifest as cytosine. The probability of this single worst-case realization is \((0.25)^3 = 1/64 \approx 0.016\). This means that \(98.4\%\) of the millions of physical strands synthesized from this sequence will have exactly 4, 5, or 6 GC nucleotides, perfectly complying with the strict 40\%--60\% thermodynamic requirements. Discarding this sequence represents a massive loss of coding capacity. 

    Allowing the sequence, on the other hand, would only slightly bias the composite pool. Therefore, the GC constraint is better evaluated probabilistically, setting the threshold on the probability $\epsilon$ of violation rather than the worst-case realization.
\end{example}

For a given composite sequence \(\mathbf{S} = [S_1, \dots, S_n] \in \mathcal{A}_k^n\), the physical manifestation of a G or C nucleotide at any position \(j \in \{1, \ldots, n\}\) is an independent Bernoulli trial. 

Let \(X_j \sim \text{Bernoulli}(q(S_j))\) represent the binary physical GC contribution of the \(j\)-th symbol. The total physical GC count of the sequence is the sum of these independently distributed random variables, denoted as \(Y = \sum_{j=1}^{n} X_j\). Consequently, \(Y\) follows a Poisson binomial distribution, allowing us to enforce a precise error tolerance \(\epsilon \in [0, 1)\).

\begin{definition} [\(\epsilon\)-Probabilistic GC Constraint]\label{def:prob_gc}
    Let \(\mathbf{S} = [S_1, \dots, S_n] \in \mathcal{A}_k^n\) be a composite sequence of length \(n\). Furthermore, let \(Y\) be the random variable depicting the total physical GC count of a DNA strand synthesized from $\mathbf{S}$. We say \(\mathbf{S}\) complies with the \(\epsilon\)-probabilistic GC constraint with limits \(\alpha \le \beta \in [0,1]\) and error tolerance $ \epsilon \in [0, 1) $ if
    \begin{equation}
        \Pr\left(Y < n\alpha\right) + \Pr\left(Y > n\beta\right) \leq \epsilon.
    \end{equation}
    We further say that \(\mathbf{S}\) complies with the \(w\)-windowed \(\epsilon\)-probabilistic GC constraint if any subsequence \(\mathbf{S}' = [S_i, \ldots , S_{i+w-1}]\) for \(i \in \{1, \ldots , n-w+1\}\) complies with the \(\epsilon\)-probabilistic GC constraint.
\end{definition}

\begin{remark}[Reduction to Deterministic Case]
For a strict error tolerance of \(\epsilon = 0\), the probabilistic approach exactly reduces to the deterministic worst-case baseline. 
\end{remark}

\section{Capacity Formulations}\label{sec:capacity}

\subsection{Global Composition Types and Error Compliance}

Evaluating the exact \ac{pmf} for every possible sequence individually is computationally intractable. We group sequences into composition types that all share the same Poisson binomial \ac{pmf}.

\begin{definition}[Sequence Type~\cite{cover_thomas_types}]
\label{def:type}
Given a sequence \(\mathbf{S} \in \mathcal{A}_k^n\), we define the type of \(\mathbf{S}\) as \(\mathbf{t}(\mathbf{S}) = (t_0(\mathbf{S}), \dots, t_k(\mathbf{S}))\), where \(t_i(\mathbf{S})\) denotes the number of positions of \(\mathbf{S}\) assigned to probability bin \(\mathcal{B}_i\) as defined in Definition \ref{def:binning}. Furthermore, given a vector \(\mathbf{t} = (t_0, \ldots , t_k)\) such that \(t_i \geq 0\) for all \(i \in \{0, \ldots , k \} \) and \(\sum_{i=0}^k t_i = n\), we denote by
\begin{equation*}
    \mathcal{T}_\mathbf{t} = \left\{ \mathbf{S} \in \mathcal{A}_k^n : \mathbf{t}(\mathbf{S}) = \mathbf{t} \right\}
\end{equation*}
the type-class of type \(\mathbf{t}\).
\end{definition}

To compute the Poisson binomial distribution for a given type, we use the fact that
the sum of independent discrete random variables has a \ac{pmf} given by the
discrete convolution of their individual \acp{pmf}~\cite{Grinstead2025Sums}.
For two independent discrete random variables \(A\) and \(B\) with \acp{pmf}
\(f\) and \(g\), the \ac{pmf} of their sum is \((f*g)(z)=\sum_x f(x)g(z-x)\).

\begin{proposition}[Exact \ac{pmf} via Convolution]
\label{prop:exact_pmf}
For a sequence of type \(\mathbf{t}=(t_0,\ldots,t_k)\), let
\(
Y_i \sim \bin\left(t_i,\frac{i}{k}\right)
\)
denote the binomially distributed GC count contributed by the symbols in bin
\(\mathcal{B}_i\), and let \(f_i\) denote the \ac{pmf} of \(Y_i\). Then the \ac{pmf} of the total GC count
\(
Y=\sum_{i=0}^{k}Y_i
\)
is
\(
f_Y=f_0*f_1*\cdots*f_k.
\)
\end{proposition}

\begin{proposition}[Type-Dependent Compliance]
\label{prop:type_compliance}
For any two composite sequences \(\mathbf{S}_1, \mathbf{S}_2 \in \mathcal{T}_\mathbf{t}\), \(\mathbf{S}_1\) complies with the global \(\epsilon\)-probabilistic GC constraint if and only if \(\mathbf{S}_2\) complies.
\end{proposition}

Because the violation probability is identical for all sequences within a given type-class, we can define the set of all probabilistically valid types. Let \(Y_\mathbf{t}\) denote the Poisson binomial random variable associated with type \(\mathbf{t}\).

\begin{definition}[Allowed Type Set]
The set of all sequence types that satisfy the global \(\epsilon\)-probabilistic GC constraint is denoted as:
\vspace{-0.4cm}
\begin{equation}
    \mathcal{P}_{\epsilon} = \left\{ \mathbf{t} : \Pr\left(Y_\mathbf{t} < n\alpha\right) + \Pr\left(Y_\mathbf{t} > n\beta\right) \le \epsilon \right\}.
\end{equation}
\end{definition}
\vspace{-0.2cm}
Because probabilistic compliance depends only on a symbol's GC probability \(q(S)\) rather than the specific nucleotide mixture of the symbol itself, the \(\epsilon\)-probabilistic GC constraint can be equivalently defined and evaluated over sequences of probability bins (e.g., \( \mathcal{B}_{i_1}, \mathcal{B}_{i_2}, \dots\) where \(\mathcal{B}_{i_j} \in \{\mathcal{B}_0, \dots, \mathcal{B}_k\}\)). 

\begin{proposition}[Type-Class Cardinality]\label{prop:class_cardinality}

For a given sequence type \(\mathbf{t} = (t_0, \dots, t_k)\), the total number of composite sequences in the type-class \(\mathcal{T}_\mathbf{t}\) is given by:
\vspace{-0.2cm}
\begin{equation}
    |\mathcal{T}_\mathbf{t}| = \frac{n!}{\prod_{i=0}^k t_i!} \prod_{i=0}^k |\mathcal{B}_i|^{t_i}.
\end{equation}
\end{proposition}

\begin{definition}
\label{def:global_GC_capacity}
Let \(\cond\) be a condition on vectors over \(\mathcal{A}_k^n\) and let \(\mathcal{V}_\cond \subseteq \mathcal{A}_k^n\) be the set of all vectors fulfilling that condition. Then we define the capacity for codes constrained by \(\cond\) as
\begin{align}
C_\cond =
\lim_{n \rightarrow \infty}
\frac{\log_2 \left| \mathcal{V}_\cond \right|}
{n}.
\end{align}
\end{definition}

\subsection{Global GC Capacity}
\label{sec:global_capacity}

\begin{theorem}[Global Probabilistic GC Capacity]
\label{prop:global_capacity}
Given a resolution parameter \(k\) and an error tolerance \(\epsilon\), let
\(\mathbf{p}_\mathbf{t} = (t_0/n, \dots, t_k/n)\) denote the empirical bin distribution of type \(\mathbf{t}\) with entropy \(H(\mathbf{p}_\mathbf{t}) = -\sum_{i=0}^k \frac{t_i}{n} \log_2 \frac{t_i}{n}\), where \(0 \log_2 0 = 0\), and let
\vspace{-0.4cm}
\begin{align*}
F(\mathbf{t}) = H(\mathbf{p}_\mathbf{t}) + \sum_{i=0}^k \frac{t_i}{n} \log_2 |\mathcal{B}_i|.
\end{align*}
Then the global \(\epsilon\)-probabilistic GC constrained capacity is
\begin{align}
    C_{\text{global-GC}}(\epsilon) 
    &= \lim_{n \rightarrow \infty} \frac{1}{n} \log_2 \left( \sum_{\mathbf{t} \in \mathcal{P}_{\epsilon}} \frac{n!}{\prod_{i=0}^k t_i!} \prod_{i=0}^k |\mathcal{B}_i|^{t_i} \right) \nonumber
    \\
    &= \lim_{n \rightarrow \infty} \max_{\mathbf{t} \in \mathcal{P}_{\epsilon}} F(\mathbf{t}),
\end{align}
where \(\mathcal{P}_{\epsilon}\) is the allowed type set for length \(n\), provided that \(\mathcal{P}_{\epsilon} \neq \emptyset\) and the limit exists.
\end{theorem}

\begin{proof}[Proof sketch]
The set \(\mathcal{V}_{\text{global-GC}}\) of valid sequences is the disjoint union of the type-classes \(\mathcal{T}_\mathbf{t}\), \(\mathbf{t} \in \mathcal{P}_\epsilon\). Each has size \(2^{nF(\mathbf{t})}\) up to a polynomial factor, and there are polynomially many types~\cite{cover_thomas_types}.
\end{proof}

Note that \(F(\mathbf{t})\) is the entropy of the distribution over \(\mathcal{A}_k\) that assigns mass \(t_i/n\) to bin \(\mathcal{B}_i\) and is uniform within each bin. Hence \(F(\mathbf{t}) \leq \log_2 |\mathcal{A}_k|\), with equality if and only if \(t_i/n = |\mathcal{B}_i|/|\mathcal{A}_k|\) for all \(i\).

\begin{remark}[Asymptotic Independence of \(\epsilon\)]
\label{rem:eps_asymptotic}
Let \(\mu(\mathbf{p}) = \sum_{i=0}^k p_i \frac{i}{k}\) denote the expected GC fraction of a sequence with bin distribution \(\mathbf{p}\). Using Chebyshev's inequality (see Appendix~\ref{app:proofs}), for \(\alpha < \beta\) and any fixed \(\epsilon \in (0,1)\), the capacity does not depend on \(\epsilon\):
\vspace{-0.3cm}
\begin{align*}
C_{\text{global-GC}}(\epsilon) = \max_{\mathbf{p} : \, \alpha \leq \mu(\mathbf{p}) \leq \beta} \left( H(\mathbf{p}) + \sum_{i=0}^k p_i \log_2 |\mathcal{B}_i| \right).
\end{align*}
In particular, since \(|\mathcal{B}_i| = |\mathcal{B}_{k-i}|\), the unconstrained maximizer \(p_i = |\mathcal{B}_i|/|\mathcal{A}_k|\) satisfies \(\mu(\mathbf{p}) = \frac{1}{2}\), and thus \(C_{\text{global-GC}}(\epsilon) = \log_2 |\mathcal{A}_k|\) whenever \(\alpha \leq \frac{1}{2} \leq \beta\). This does not hold for the deterministic case \(\epsilon = 0\), see Corollary~\ref{cor:det_capacity} in Appendix~\ref{app:deterministic}.
\end{remark}

\subsection{Sliding-Window GC Capacity}
\label{sec:window_capacity}

When the GC-content constraint is enforced locally over a sliding window of length \(w<n\), sequence validity depends on overlapping windows. Under the \(\epsilon\)-probabilistic constraint, the probability distribution of the GC count depends on the probability bins of the symbols within each window. We therefore model the constrained sequence as a finite-state Markov chain, where each state stores the recent history needed to determine whether the next window is admissible. An example is given in Appendix~\ref{app:examples}.

\begin{theorem}[Probabilistic Sliding-Window GC Capacity]
\label{thm:window_prob_capacity}
Let \(\mathcal{G}(V,E)\) be a directed graph whose vertices \(V\) are the set of all \((w-1)\)-tuples of probability-mass bins \(\mathcal{B}_i\), \(i\in\{0,\ldots,k\}\). 
An edge from a vertex \(s=(b_1,\ldots,b_{w-1})\) to a vertex \(s'=(b'_1,\ldots,b'_{w-1})\) exists if the following two conditions hold: 
(i) \((b_2,\ldots,b_{w-1})=(b'_1,\ldots,b'_{w-2})\), and
(ii) the \(w\)-tuple \((b_1,\ldots,b_{w-1},b'_{w-1})\) satisfies the \(\epsilon\)-probabilistic GC constraint. Each valid edge corresponding to the appended bin \(\mathcal{B}_{j}\) is assigned weight \(|\mathcal{B}_{j}|\). 
Let \(T_\epsilon\) denote the resulting weighted transition matrix and let \(v^{(w)}_\epsilon(s)\) count the valid length-\(w\) composite sequences whose last \(w-1\) bins equal \(s\), i.e., \(v^{(w)}_\epsilon = u T_\epsilon\) with \(u(s) = \prod_{j=1}^{w-1} |\mathcal{B}_{b_j}|\). 
Then, for \(n \geq w\), the number of elements in the set \(\mathcal{V}_{\mathrm{w-GC}}\) of composite sequences fulfilling the \(w\)-windowed \(\epsilon\)-probabilistic GC constraint is
\begin{equation}
    |\mathcal{V}_{\mathrm{w-GC}}|
    =
    v^{(w)}_\epsilon T_\epsilon^{\,n-w}\mathbf{1}^{T},
\end{equation}
and, if the largest absolute eigenvalue of \(T_\epsilon\) is positive, i.e., \(\lambda(T_\epsilon) > 0\), the corresponding capacity is
\begin{equation}
    C_{\mathrm{w-GC}}(\epsilon)
    =
    \lim_{n \rightarrow \infty}
    \frac{1}{n} \log_2 |\mathcal{V}_{\mathrm{w-GC}}|
    =
    \log_2 \lambda(T_\epsilon).
\end{equation}
\end{theorem}

\begin{proof}[Proof sketch]
Valid sequences correspond to weighted walks in \(\mathcal{G}\), and the growth rate of \(T_\epsilon^{\,n-w}\) is \(\lambda(T_\epsilon)\) by Gelfand's formula.
\end{proof}

\begin{remark}[Role of \(\epsilon\) for Sliding Windows]
\label{rem:eps_window}
For \(k=4\), \(\alpha=0.4\), and \(\beta=0.6\), evaluating \(\lambda(T_\epsilon)\) gives normalized capacities \(C_{\mathrm{w-GC}}(\epsilon)/\log_2|\mathcal{A}_4| \approx 0.629, 0.629, 0.770\) for \(w=5\) and \(\epsilon = 0, 0.05, 0.2\), and \( C_{\mathrm{w-GC}}(\epsilon)/\log_2|\mathcal{A}_4| \approx 0.674, 0.678, 0.839\) for \(w=10\). In contrast to the global constraint (Remark~\ref{rem:eps_asymptotic}), \(\epsilon\) determines the capacity, since the window length \(w\) does not grow with $n$. For small \(\epsilon\), it coincides with the deterministic capacity, see Appendix~\ref{app:deterministic}.
\end{remark}

\subsection{Analysis of Results}

\begin{figure}[htpb]
    \centering
    \begin{tikzpicture}
    \begin{axis}[
        width=0.95\columnwidth,
        height=0.72\columnwidth,
        xlabel={Strand Length, \(n\)},
        ylabel={Finite-Length Rate (bits/nt)},
        xmin=10, xmax=100,
        grid=major,
        grid style={gray!30},
        legend style={
            at={(0.5,-0.22)}, % Konum 2 satır için yukarı çekildi
            anchor=north,
            draw=none,
            font=\footnotesize,
            /tikz/every even column/.append style={
                column sep=0.15cm % Sütunlar arası mesafe daraltıldı
            }
        },
        legend columns=4, % Tekrar 4 sütuna dönüldü
        tick align=outside,
    ]

    % =========================================================
    % UNCONSTRAINED LIMIT
    % =========================================================

    \addplot[
        black,
        dash dot,
        line width=1.1pt,
        forget plot
    ]
    coordinates {
        (10,5.1293) (100,5.1293)
    };

    % =========================================================
    % EPSILON = 0
    % =========================================================

    % Global
    \addplot[
        epszero,
        solid,
        line width=1pt,
        mark=*,
        mark size=2pt,
        forget plot
    ]
    table[
        x index=0,
        y=Global_eps_0,
        col sep=comma
    ] {figures/final_plot_data.csv};

    % Sliding window w = 5
    \addplot[
        epszero,
        densely dotted,
        line width=1pt,
        mark=*,
        mark size=2pt,
        forget plot
    ]
    table[
        x index=0,
        y=SW_w5_eps_0,
        col sep=comma
    ] {figures/final_plot_data.csv};

    % Sliding window w = 10
    \addplot[
        epszero,
        densely dashed,
        line width=1pt,
        mark=*,
        mark size=2pt,
        forget plot
    ]
    table[
        x index=0,
        y=SW_w10_eps_0,
        col sep=comma
    ] {figures/final_plot_data.csv};

    % =========================================================
    % EPSILON = 0.05
    % =========================================================

    % Global
    \addplot[
        epsfive,
        solid,
        line width=1pt,
        mark=square*,
        mark size=2pt,
        forget plot
    ]
    table[
        x index=0,
        y=Global_eps_0p05,
        col sep=comma
    ] {figures/final_plot_data.csv};

    % Sliding window w = 5
    \addplot[
        epsfive,
        densely dotted,
        line width=1pt,
        mark=square*,
        mark size=2pt,
        forget plot
    ]
    table[
        x index=0,
        y=SW_w5_eps_0p05,
        col sep=comma
    ] {figures/final_plot_data.csv};

    % Sliding window w = 10
    \addplot[
        epsfive,
        densely dashed,
        line width=1pt,
        mark=square*,
        mark size=2pt,
        forget plot
    ]
    table[
        x index=0,
        y=SW_w10_eps_0p05,
        col sep=comma
    ] {figures/final_plot_data.csv};

    % =========================================================
    % EPSILON = 0.20
    % =========================================================

    % Global
    \addplot[
        epstwenty,
        solid,
        line width=1pt,
        mark=diamond*,
        mark size=2pt,
        forget plot
    ]
    table[
        x index=0,
        y=Global_eps_0p2,
        col sep=comma
    ] {figures/final_plot_data.csv};

    % Sliding window w = 5
    \addplot[
        epstwenty,
        densely dotted,
        line width=1pt,
        mark=diamond*,
        mark size=2pt,
        forget plot
    ]
    table[
        x index=0,
        y=SW_w5_eps_0p2,
        col sep=comma
    ] {figures/final_plot_data.csv};

    % Sliding window w = 10
    \addplot[
        epstwenty,
        densely dashed,
        line width=1pt,
        mark=diamond*,
        mark size=2pt,
        forget plot
    ]
    table[
        x index=0,
        y=SW_w10_eps_0p2,
        col sep=comma
    ] {figures/final_plot_data.csv};

    % =========================================================
    % LEGEND ENTRIES (4 Columns)
    % =========================================================

    \addlegendimage{black, solid, line width=1pt}
    \addlegendentry{Global}

    \addlegendimage{black, densely dotted, line width=1pt}
    \addlegendentry{SW (\(w=5\))}

    \addlegendimage{black, densely dashed, line width=1pt}
    \addlegendentry{SW (\(w=10\))}

    \addlegendimage{black, dash dot, line width=1pt}
    \addlegendentry{Unc.}

    \addlegendimage{epszero, solid, mark=*, mark size=2pt, line width=1pt}
    \addlegendentry{\(\epsilon=0\)}

    \addlegendimage{epsfive, solid, mark=square*, mark size=2pt, line width=1pt}
    \addlegendentry{\(\epsilon=0.05\)}

    \addlegendimage{epstwenty, solid, mark=diamond*, mark size=2pt, line width=1pt}
    \addlegendentry{\(\epsilon=0.20\)}

    \end{axis}
\end{tikzpicture}
    \caption{
    Global and sliding-window finite-length rates as a function of strand
    length \(n\), for \(k=4\), \(\alpha=0.40\), and \(\beta=0.60\). The rates are
    shown for error tolerances \(\epsilon\in\{0,0.05,0.20\}\) and sliding-window
    lengths \(w\in\{5,10\}\). (SW: Sliding Window, Unc.: Unconstrained Limit)}
    \label{fig:capacity_analysis}
\end{figure}
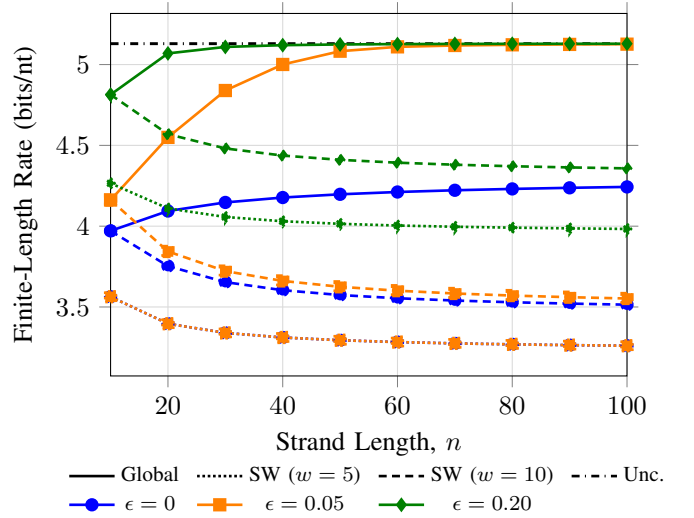

Figure~\ref{fig:capacity_analysis} evaluates the global and sliding-window finite-length rates as a function of strand length for a resolution parameter of $k = 4$. The global rate increases with strand length and converges toward the unconstrained capacity if $\epsilon > 0$. 
Conversely, the sliding-window rate strictly decreases as strand length grows since each newly appended symbol must simultaneously satisfy the GC bounds across multiple shifting, overlapping windows.

\section{Multi-Type Arithmetic Indexing for the Probabilistic GC-Constraint}\label{sec:coding_schemes}

To achieve the theoretical capacity limits derived in Section~\ref{sec:global_capacity} for the global \(\epsilon\)-probabilistic GC constraint, we leverage the global composition types to implement a capacity-achieving multi-type arithmetic encoder.

Our approach adapts the arithmetic coding architecture used in constant-composition~\cite{CCD} and multi-composition~\cite{MCDM} distribution matching. Our objective, however, is not to shape a distribution, but to use \textit{all} sequences whose composition types belong to the probabilistically valid set \(\mathcal{P}_\epsilon\). We therefore utilize the multi-composition arithmetic coding framework as a capacity-achieving enumerative indexer.

To formalize this for the probabilistic GC constraint, we rely on standard interval-partitioning machinery~\cite{CCD, MCDM}. A binary input data sequence \(\mathbf{b}\) of length \(m = \lfloor \log_2 |\mathcal{V}_{\text{global-GC}}| \rfloor\) is converted into a target fraction \(0\leq d < 1\) via its \ac{nbc} representation, \(d = \operatorname{NBC}(\mathbf{b}) / 2^m\). The arithmetic encoder recursively partitions this interval such that every valid composite sequence \(\mathbf{S} \in \mathcal{V}_{\text{global-GC}}\) corresponds to a disjoint subinterval \(I(\mathbf{S})\) of strictly equal length \(1/|\mathcal{V}_{\text{global-GC}}|\). This equal-length partitioning guarantees that the maximum number of messages can be uniquely mapped and recovered~\cite{CCD, MCDM}.

Applying this continuous interval-partitioning to our discrete constrained channel requires two problem-specific ingredients:
(i) a defined lexicographic ordering over our composite alphabet \(\mathcal{A}_k\) to consistently assign interval boundaries, and (ii) a method to dynamically compute the conditional branching probability \(P(a|\Tilde{\mathbf{S}})\) that a valid sequence continues with symbol \(a \in \mathcal{A}_k\), given an existing prefix \(\Tilde{\mathbf{S}}\).

To maintain equal-length final intervals, the branching probability must exactly reflect the fraction of valid codewords that extend the current prefix. Specifically, let \(N(\Tilde{\mathbf{S}})\) represent the number of complete valid codewords that share the prefix \(\Tilde{\mathbf{S}}\). The conditional probability of appending symbol \(a\) is given by the ratio \(P(a|\Tilde{\mathbf{S}}) = N(\Tilde{\mathbf{S}}a)/N(\Tilde{\mathbf{S}})\)~\cite{MCDM}. To compute this ratio dynamically without storing the massive multi-type codebook, we introduce the unified prefix-completion count.

\begin{proposition}[Unified Prefix-Completion Count]
\label{prop:prefix_count}
Let \(\Tilde{\mathbf{S}}\) be a prefix of length \(\ell\). For any allowed type \(\mathbf{t} = (t_0, \dots, t_k) \in \mathcal{P}_\epsilon\), let \(c_i(\Tilde{\mathbf{S}})\) denote the number of symbols currently in \(\Tilde{\mathbf{S}}\) that belong to probability bin \(\mathcal{B}_i\). The remaining number of symbols required from bin \(\mathcal{B}_i\) to complete type \(\mathbf{t}\) is \(t_i(\Tilde{\mathbf{S}}) = t_i - c_i(\Tilde{\mathbf{S}})\). The total number of valid length-\(n\) sequences extending the prefix \(\Tilde{\mathbf{S}}\) is:
\begin{equation}
    N(\Tilde{\mathbf{S}}) = \sum_{\mathbf{t} \in \mathcal{P}_{\epsilon}} N_\mathbf{t}(\Tilde{\mathbf{S}}),
\end{equation}
where the per-type continuation count \(N_\mathbf{t}(\Tilde{\mathbf{S}})\) is given by the multinomial extension:
\begin{equation}
    N_\mathbf{t}(\Tilde{\mathbf{S}}) \! = \!\!
    \begin{cases} 
        \frac{(n-\ell)!}{\prod_{i=0}^{k} t_i(\Tilde{\mathbf{S}})!} \prod\limits_{i=0}^{k} |\mathcal{B}_i|^{t_i(\Tilde{\mathbf{S}})}, 
        & \!\!\!\! \text{if } t_i(\Tilde{\mathbf{S}}) \! \ge \! 0 \text{ for all } i, \!\!
        \\
        0, & \!\!\!\! \text{otherwise}.
    \end{cases}
\end{equation}
\end{proposition}

Using the prefix-completion counts derived from Proposition~\ref{prop:prefix_count}, the encoder recursively narrows the initial interval $[0, 1)$ around the target fraction $d$. This symbol-by-symbol mapping process is formalized in Algorithm~\ref{alg:encoder}. A worked example is given in Appendix~\ref{app:examples}.

\begin{figure}[!t]
\noindent\begin{minipage}{\columnwidth}
\rule{\columnwidth}{1pt}
\refstepcounter{algorithm}\label{alg:encoder}
\textbf{Algorithm~\thealgorithm: Multi-Type Arithmetic-Indexing Encoder} \\
\rule{\columnwidth}{1pt}
\textbf{Input:} Binary message \(\mathbf{b}\), strand length \(n\), allowed types \(\mathcal{P}_{\epsilon}\). \\
\textbf{Output:} Admissible composite codeword \(\mathbf{c}\).
\begin{algorithmic}[1]
    \STATE \(d \leftarrow \operatorname{NBC}(\mathbf{b})/2^m\) \hfill \{\textit{Target fraction}\}
    \STATE \(I \leftarrow [0, 1)\) \hfill \{\textit{Current interval}\}
    \STATE \(s \leftarrow \emptyset\) \hfill \{\textit{Initialize prefix}\}
    \FOR{\(i=1\) \TO \(n\)}
        \STATE For each valid candidate \(a\), compute multi-type paths:
        \STATE \(N(sa) = \sum_{\mathbf{t} \in \mathcal{P}_{\epsilon}} N_\mathbf{t}(sa)\)
        \STATE Partition \(I\) into sub-intervals \(I_a\) proportional to \(N(sa)\).
        \STATE Select symbol \(a^*\) whose sub-interval \(I_{a^*}\) contains \(d\).
        \STATE \(I \leftarrow I_{a^*}\) \hfill \{\textit{Update interval}\}
        \STATE \(s \leftarrow sa^*\) \hfill \{\textit{Append to prefix}\}
    \ENDFOR
    \RETURN \(\mathbf{c} \leftarrow s\)
\end{algorithmic}
\rule{\columnwidth}{1pt}
\end{minipage}
\end{figure}

To decode the codeword, the receiver applies arithmetic compression on the received sequence \(\mathbf{c}\) using the same branching model \(P(a|\Tilde{\mathbf{S}}) \propto N(\Tilde{\mathbf{S}}a)\). The decoder recursively partitions the interval \(I\) symbol by symbol to reconstruct the final valid sub-interval \(I(\mathbf{S})\). Retrieving the original binary sequence \(\mathbf{b}\) is performed by extracting the unique data point \(d\) contained within \(I(\mathbf{S})\), and reversing the binary mapping via \(\operatorname{NBC}^{-1}\)~\cite{MCDM}.

To measure the practical efficiency of this encoding scheme, we calculate its normalized achieved information rate, defined as the fraction of the unconstrained composite-symbol capacity achieved by the practical encoder:
\begin{equation}
    \eta_{\mathrm{enc}} = \frac{\left\lfloor\log_2|\mathcal{V}_{\text{global-GC}}|\right\rfloor}{n\log_2|\mathcal{A}_k|}.
\end{equation}

\begin{figure}
    \centering
    \begin{tikzpicture}
    % Define the colors if they aren't already in your preamble
    \colorlet{epszero}{blue}
    \colorlet{epsfive}{orange}
    \colorlet{epstwenty}{green!70!black}

    \begin{axis}[
        width=0.95\columnwidth,
        height=0.6\columnwidth,
        xlabel={Strand Length, \(n\)},
        ylabel={Normalized Rate / Capacity},
        xmin=10, xmax=60,
        ymin=0.6, ymax=1,
        ytick={0.6,0.7,...,1},
        grid=major,
        grid style={gray!30},
        legend style={
            at={(0.98,0.04)}, % Placed inside the plot, bottom right
            anchor=south east,
            draw=none,
            fill=white,       % Adds a white background to cover grid lines
            fill opacity=0.9, % Slight transparency looks elegant
            text opacity=1,
            font=\footnotesize,
            /tikz/every even column/.append style={
                column sep=0.35cm
            }
        },
        legend columns=3,
        tick align=outside,
        axis lines*=left,     % Only left and bottom axis lines
        axis line style={-},  % No arrow tips
    ]

    % =========================================================
    % UNCONSTRAINED BASELINE
    % =========================================================
    \addplot[
        black,
        dash dot,
        line width=1.1pt,
        forget plot
    ]
    coordinates {(10,1) (60,1)};

    % =========================================================
    % THEORETICAL CAPACITY (Dashed lines, no markers)
    % =========================================================
    
    \addplot[epszero, densely dashed, line width=1pt, forget plot]
    table[x=w, y=Cap_eps_0, col sep=comma] {figures/capacity_efficiency_data.csv};

    \addplot[epsfive, densely dashed, line width=1pt, forget plot]
    table[x=w, y=Cap_eps_0p05, col sep=comma] {figures/capacity_efficiency_data.csv};

    \addplot[epstwenty, densely dashed, line width=1pt, forget plot]
    table[x=w, y=Cap_eps_0p2, col sep=comma] {figures/capacity_efficiency_data.csv};

    % =========================================================
    % MULTI-TYPE ARITHMETIC INDEXING (Solid lines, filled markers)
    % =========================================================
    
    \addplot[epszero, solid, line width=1pt, mark=*, mark size=2pt, forget plot]
    table[x=w, y=MTAI_eps_0, col sep=comma] {figures/capacity_efficiency_data.csv};

    \addplot[epsfive, solid, line width=1pt, mark=square*, mark size=2pt, forget plot]
    table[x=w, y=MTAI_eps_0p05, col sep=comma] {figures/capacity_efficiency_data.csv};

    \addplot[epstwenty, solid, line width=1pt, mark=diamond*, mark size=2pt, forget plot]
    table[x=w, y=MTAI_eps_0p2, col sep=comma] {figures/capacity_efficiency_data.csv};

    % =========================================================
    % LEGEND ENTRIES (3 Columns x 2 Rows)
    % =========================================================

    % Row 1: Line styles / Quantities
    \addlegendimage{black, densely dashed, line width=1pt}
    \addlegendentry{Capacity}

    \addlegendimage{black, solid, mark=*, mark size=2pt, line width=1pt}
    \addlegendentry{MTAI}

    \addlegendimage{black, dash dot, line width=1pt}
    \addlegendentry{Unconstrained}

    % Row 2: Colors / Epsilons
    \addlegendimage{epszero, solid, mark=*, mark size=2pt, line width=1pt}
    \addlegendentry{\(\epsilon=0\)}

    \addlegendimage{epsfive, solid, mark=square*, mark size=2pt, line width=1pt}
    \addlegendentry{\(\epsilon=0.05\)}

    \addlegendimage{epstwenty, solid, mark=diamond*, mark size=2pt, line width=1pt}
    \addlegendentry{\(\epsilon=0.20\)}

    \end{axis}
\end{tikzpicture}
    \caption{Capacity and coding efficiency comparison as a function of strand length \(n\), for \(k=4\), \(\alpha=0.40\), and \(\beta=0.60\). Rates are shown for error tolerances \(\epsilon\in\{0,0.05,0.20\}\).}
    \label{fig:coding_efficiency}
\end{figure}
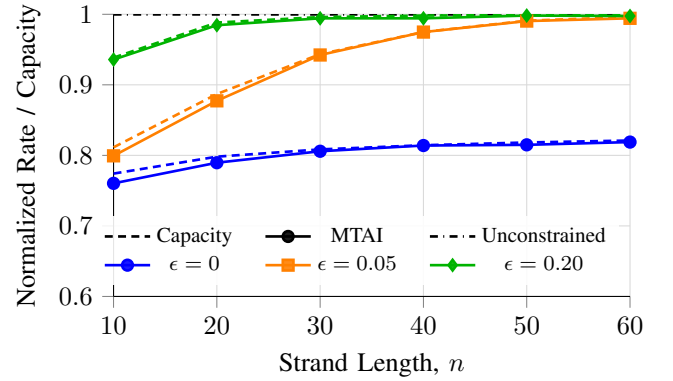

As illustrated in Figure~\ref{fig:coding_efficiency}, the multi-type arithmetic indexing (MTAI) approach consistently converges to the theoretical GC-constrained channel capacity already at short strand lengths across all evaluated error tolerances \(\epsilon\).

\section{Conclusion}

This paper introduced a probabilistic framework for GC-content constraints in composite DNA storage. By grouping composite symbols into probability bins, we formulated exact capacity bounds for global constraints using composition types, and for local sliding-window constraints via finite-state transfer matrices. Furthermore, we proposed a capacity-achieving multi-type enumerative encoder that efficiently utilizes this expanded code space by directly indexing the probabilistically valid codebook. Future work will explore adapting this arithmetic indexing scheme directly for sliding-window constraints.

\section*{Acknowledgements}
The authors used the AI system Gemini (Google) to assist in improving the code for the simulations and generate the figures in this manuscript. 
All AI-generated content was checked for accuracy and revised by the authors, who take full responsibility for the content of this paper.

\bibliographystyle{IEEEtran} 
\bibliography{my_references}

\begin{thebibliography}{10}
\providecommand{\url}[1]{#1}
\csname url@samestyle\endcsname
\providecommand{\newblock}{\relax}
\providecommand{\bibinfo}[2]{#2}
\providecommand{\BIBentrySTDinterwordspacing}{\spaceskip=0pt\relax}
\providecommand{\BIBentryALTinterwordstretchfactor}{4}
\providecommand{\BIBentryALTinterwordspacing}{\spaceskip=\fontdimen2\font plus
\BIBentryALTinterwordstretchfactor\fontdimen3\font minus \fontdimen4\font\relax}
\providecommand{\BIBforeignlanguage}[2]{{%
\expandafter\ifx\csname l@#1\endcsname\relax
\typeout{** WARNING: IEEEtran.bst: No hyphenation pattern has been}%
\typeout{** loaded for the language `#1'. Using the pattern for}%
\typeout{** the default language instead.}%
\else
\language=\csname l@#1\endcsname
\fi
#2}}
\providecommand{\BIBdecl}{\relax}
\BIBdecl

\bibitem{Church2012DNAStorage}
\BIBentryALTinterwordspacing
G.~M. Church, Y.~Gao, and S.~Kosuri, ``Next-generation digital information storage in dna,'' \emph{Science}, vol. 337, no. 6102, pp. 1628--1628, 2012. [Online]. Available: \url{https://www.science.org/doi/abs/10.1126/science.1226355}
\BIBentrySTDinterwordspacing

\bibitem{goldman2013}
\BIBentryALTinterwordspacing
N.~Goldman, P.~Bertone, S.~Chen, C.~Dessimoz, E.~M. LeProust, B.~Sipos, and E.~Birney, ``Towards practical, high-capacity, low-maintenance information storage in synthesized dna,'' \emph{Nature}, vol. 494, no. 7435, pp. 77--80, 2013. [Online]. Available: \url{https://doi.org/10.1038/nature11875}
\BIBentrySTDinterwordspacing

\bibitem{sabary2024reconstruction}
O.~Sabary, A.~Yucovich, G.~Shapira, and E.~Yaakobi, ``Reconstruction algorithms for dna-storage systems,'' \emph{Scientific Reports}, vol.~14, no.~1, p. 1951, 2024.

\bibitem{anavy2019}
\BIBentryALTinterwordspacing
L.~Anavy, I.~Vaknin, O.~Atar, R.~Amit, and Z.~Yakhini, ``Data storage in dna with fewer synthesis cycles using composite dna letters,'' \emph{Nature Biotechnology}, vol.~37, no.~10, pp. 1229--1236, 2019. [Online]. Available: \url{https://doi.org/10.1038/s41587-019-0240-x}
\BIBentrySTDinterwordspacing

\bibitem{Erlich2017DNAFountain}
Y.~Erlich and D.~Zielinski, ``Dna fountain enables a robust and efficient storage architecture,'' \emph{Science}, vol. 355, no. 6328, pp. 950--954, 2017.

\bibitem{Heckel2019Characterization}
R.~Heckel, G.~Mikutis, and R.~N. Grass, ``A characterization of the dna data storage channel,'' \emph{Scientific Reports}, vol.~9, no.~1, p. 9663, 2019.

\bibitem{Yazdi2016DNAStorage}
S.~H.~T. Yazdi, H.~M. Kiah, E.~Garcia-Ruiz, J.~Ma, H.~Zhao, and O.~Milenkovic, ``Dna-based storage: Trends and methods,'' \emph{IEEE Transactions on Molecular, Biological, and Multi-Scale Communications}, vol.~1, no.~3, pp. 230--248, 2016.

\bibitem{optical_constrained}
K.~A.~S. Immink, P.~H. Siegel, and J.~K. Wolf, \emph{Codes for digital recorders (invited paper)}.\hskip 1em plus 0.5em minus 0.4em\relax IEEE Press, 2000, p. 216–255.

\bibitem{nguyen}
T.~T. Nguyen, C.~Wang, K.~Cai, Y.~Zhang, and Z.~Yakhini, ``Constrained coding for composite dna: Channel capacity and efficient constructions,'' in \emph{2025 IEEE International Symposium on Information Theory (ISIT)}, 2025, pp. 1--6.

\bibitem{walter2025coding}
F.~Walter and Y.~Yehezkeally, ``Coding for strand breaks in composite dna,'' in \emph{2025 IEEE International Symposium on Information Theory (ISIT)}, 2025, pp. 1--6.

\bibitem{dang2026rll}
D.~H. Dang and S.~Kim, ``A theoretical framework for run-length-limited coding in composite dna,'' \emph{IEEE Transactions on Molecular, Biological, and Multi-Scale Communications}, 2026.

\bibitem{cover_thomas_types}
T.~M. Cover, \emph{Elements of information theory}.\hskip 1em plus 0.5em minus 0.4em\relax John Wiley \& Sons, 1999.

\bibitem{Grinstead2025Sums}
C.~M. Grinstead and J.~Laurie~Snell, ``Sums of {Discrete} {Random} {Variables},'' sep 13 2025, [Online; accessed 2026-09-15].

\bibitem{CCD}
P.~Schulte and G.~Böcherer, ``Constant composition distribution matching,'' \emph{IEEE Transactions on Information Theory}, vol.~62, no.~1, pp. 430--434, 2016.

\bibitem{MCDM}
M.~Pikus and W.~Xu, ``Arithmetic coding based multi-composition codes for bit-level distribution matching,'' in \emph{2019 IEEE Wireless Communications and Networking Conference (WCNC)}, 2019, pp. 1--6.

\end{thebibliography}

\clearpage
\appendices

\section{Proofs}\label{app:proofs}

\begin{proof}[Proof of Proposition~\ref{prop:exact_pmf}]
A type-\(\mathbf{t}\) sequence contains \(t_i\) symbols from
\(\mathcal{B}_i\). Their GC contributions are independent Bernoulli
trials with success probability \(\frac{i}{k}\), so
\(Y_i \sim \bin(t_i,\frac{i}{k})\). Since the variables
\(Y_0,\ldots,Y_k\) are mutually independent and
\(Y=\sum_{i=0}^k Y_i\), the \ac{pmf} of \(Y\) is the convolution of their
individual \acp{pmf}~\cite{Grinstead2025Sums}.
\end{proof}

\begin{proof}[Proof of Proposition~\ref{prop:type_compliance}]
As established in Proposition~\ref{prop:exact_pmf}, the exact Poisson binomial \ac{pmf} of a sequence's total physical GC count \(Y\) is determined entirely by the convolution of the binomial distributions associated with each probability bin. These binomial distributions depend exclusively on the bin counts \((t_0, \dots, t_k)\) specified by the type \(\mathbf{t}\). Since \(\mathbf{S}_1\) and \(\mathbf{S}_2\) belong to the same type-class \(\mathcal{T}_\mathbf{t}\), their total GC counts are identically distributed. Consequently, their tail violation probabilities \(\Pr\left(Y < n\alpha\right) + \Pr\left(Y > n\beta\right)\) are exactly equal.
\end{proof}

\begin{proof}[Proof of Proposition~\ref{prop:class_cardinality}]
Constructing a sequence of type \(\mathbf{t}\) is a two-step process. First, we must assign each of the \(n\) sequence positions to one of the \(k+1\) probability bins such that exactly \(t_i\) positions are assigned to bin \(\mathcal{B}_i\). The number of distinct ways to partition these positions is given by the multinomial coefficient \(\frac{n!}{\prod_{i=0}^k t_i!}\). 

Second, for each of the \(t_i\) positions assigned to bin \(\mathcal{B}_i\), we must select a specific composite symbol from that bin. Since there are \(|\mathcal{B}_i|\) distinct symbols in bin \(\mathcal{B}_i\), there are \(|\mathcal{B}_i|^{t_i}\) possible symbol combinations for those positions and the proposition follows. 
\end{proof}

\begin{proof}[Proof of Theorem~\ref{prop:global_capacity}]
Let \(\mathcal{V}_{\text{global-GC}}\) denote the set of valid sequences satisfying the global probabilistic GC constraint. 
A sequence satisfies the constraint if and only if its type belongs to \(\mathcal{P}_{\epsilon}\). Consequently, \(\mathcal{V}_{\text{global-GC}}\) forms the disjoint union of valid type-classes, \(\mathcal{V}_{\text{global-GC}} = \bigcup_{\mathbf{t} \in \mathcal{P}_{\epsilon}} \mathcal{T}_\mathbf{t}\). We obtain the total number of valid sequences \(|\mathcal{V}_{\text{global-GC}}|\) by summing their cardinalities via Proposition \ref{prop:class_cardinality}. 

From \cite[Theorem 11.1.3]{cover_thomas_types}, applied to the sequence of probability bins over an alphabet of size \(k+1\), we know that
\begin{align*}
   \frac{1}{(n+1)^{k+1}} 2^{n H(\mathbf{p}_\mathbf{t})} \leq \frac{n!}{\prod_{i=0}^k t_i!} \leq 2^{n H(\mathbf{p}_\mathbf{t})}.
\end{align*}
Since \(\prod_{i=0}^k |\mathcal{B}_i|^{t_i} = 2^{n \sum_{i} \frac{t_i}{n} \log_2 |\mathcal{B}_i|}\), Proposition~\ref{prop:class_cardinality} gives
\begin{align*}
   \frac{1}{(n+1)^{k+1}} 2^{n F(\mathbf{t})} \leq \left| \mathcal{T}_\mathbf{t} \right| \leq 2^{n F(\mathbf{t})}.
\end{align*}
Let \(F^*_n = \max_{\mathbf{t} \in \mathcal{P}_{\epsilon}} F(\mathbf{t})\). As the number of types is at most \((n+1)^{k+1}\)~\cite[Theorem 11.1.1]{cover_thomas_types}, we have
\begin{align*}
   \frac{1}{(n+1)^{k+1}} 2^{n F^*_n} \leq |\mathcal{V}_{\text{global-GC}}| \leq (n+1)^{k+1} 2^{n F^*_n}.
\end{align*}
Taking \(\frac{1}{n}\log_2(\cdot)\), both bounds differ from \(F^*_n\) by at most \(\frac{(k+1)\log_2(n+1)}{n} \rightarrow 0\) as \(n \rightarrow \infty\), which proves the claim.
\end{proof}

\begin{proof}[Proof of Remark~\ref{rem:eps_asymptotic}]
The Poisson binomial variable \(Y_\mathbf{t}\) has mean \(n\mu(\mathbf{p}_\mathbf{t})\) and variance at most \(n/4\). By Chebyshev's inequality, every type with
\begin{align*}
\alpha + \frac{1}{2\sqrt{n\epsilon}} \leq \mu(\mathbf{p}_\mathbf{t}) \leq \beta - \frac{1}{2\sqrt{n\epsilon}}
\end{align*}
belongs to \(\mathcal{P}_{\epsilon}\), while for any fixed \(\delta > 0\) and sufficiently large \(n\), no type with \(\mu(\mathbf{p}_\mathbf{t}) < \alpha - \delta\) or \(\mu(\mathbf{p}_\mathbf{t}) > \beta + \delta\) does. Since \(F\) is continuous, the claim follows from Theorem~\ref{prop:global_capacity}.
\end{proof}

\begin{proof}[Proof of Theorem~\ref{thm:window_prob_capacity}]
By construction, \(v^{(w)}_{\epsilon}\) counts valid length-\(w\)
sequences, while \(T_{\epsilon}\) extends them through valid weighted
transitions, with each transition weight accounting for the number of
composite symbols in the appended bin. Hence
\[
|\mathcal{V}_{\mathrm{w-GC}}|
=
v^{(w)}_{\epsilon}T_{\epsilon}^{\,n-w}\mathbf{1}^{T}.
\]
Since sequences shorter than \(w\) are unconstrained, every length-\((w-1)\) prefix is valid, which gives \(v^{(w)}_\epsilon = u T_\epsilon\). Let \(m = n-w+1\) and let \(\sigma_m = \mathbf{1} T_\epsilon^{m} \mathbf{1}^T\) be the sum of all entries of \(T_\epsilon^m\). As \(1 \leq u(s) \leq |\mathcal{A}_k|^{w-1}\) for all \(s\), we have \(\sigma_m \leq |\mathcal{V}_{\mathrm{w-GC}}| \leq |\mathcal{A}_k|^{w-1} \sigma_m\). Since \(T_\epsilon\) is non-negative, \(\|T_\epsilon^m\|_\infty \leq \sigma_m \leq (k+1)^{w-1} \|T_\epsilon^m\|_\infty\), where \(\|\cdot\|_\infty\) is the maximum row sum and \((k+1)^{w-1}\) is the number of states. By Gelfand's formula, \(\|T_\epsilon^m\|_\infty^{1/m} \rightarrow \lambda(T_\epsilon)\), and since \(|\mathcal{A}_k|^{w-1}\) and \((k+1)^{w-1}\) do not depend on \(n\), we obtain \(\frac{1}{n}\log_2 |\mathcal{V}_{\mathrm{w-GC}}| \rightarrow \log_2 \lambda(T_\epsilon)\).
\end{proof}

\begin{proof}[Proof of Proposition~\ref{prop:prefix_count}]
Because the type-classes of the allowed types in \(\mathcal{P}_\epsilon\) are mutually disjoint, the total number of valid extensions \(N(\Tilde{\mathbf{S}})\) is exactly the sum of the extensions belonging to each valid type, \(N(\Tilde{\mathbf{S}}) = \sum_{\mathbf{t} \in \mathcal{P}_{\epsilon}} N_\mathbf{t}(\Tilde{\mathbf{S}})\).

To evaluate the extensions for a specific type \(\mathbf{t}\), observe that the prefix \(\Tilde{\mathbf{S}}\) of length \(\ell\) has already consumed \(c_i(\Tilde{\mathbf{S}})\) symbols from each bin \(\mathcal{B}_i\). Therefore, completing a sequence of exactly type \(\mathbf{t}\) requires filling the remaining \(n - \ell\) positions such that exactly \(t_i(\Tilde{\mathbf{S}}) = t_i - c_i(\Tilde{\mathbf{S}})\) positions are assigned to bin \(\mathcal{B}_i\) for every \(i \in \{0, \ldots, k\}\).

If \(t_i(\Tilde{\mathbf{S}}) \leq 0\) for any bin, the prefix has already exceeded the allowed count for that type, yielding $0$ valid extensions. Otherwise, applying the combinatorial logic established in Proposition~\ref{prop:class_cardinality} 
%(Type-Class Cardinality) 
to the remaining \(n - \ell\) positions yields the product of the multinomial partition coefficient \(\frac{(n-\ell)!}{\prod_{i=0}^{k} t_i(\Tilde{\mathbf{S}})!}\) and the independent symbol combinations \(\prod_{i=0}^{k} |\mathcal{B}_i|^{t_i(\Tilde{\mathbf{S}})}\).
\end{proof}

\section{Deterministic Baseline Capacities}\label{app:deterministic}

\begin{corollary}[Deterministic Global GC Capacity]
\label{cor:det_capacity}
Let \(k \geq 2\), \(\alpha \leq \beta\), and let \(m_k = |\mathcal{A}_k| - 2(k+1)\) denote the number of mixed symbols, i.e., symbols in \(\mathcal{B}_1, \dots, \mathcal{B}_{k-1}\). With \(p_m = 1 - p_0 - p_k\), define
\begin{align}
G(p_0, p_k) = H(p_0, p_k, p_m) &+ (p_0 + p_k) \log_2 (k+1)
\nonumber \\
&+ p_m \log_2 m_k.
\label{eq:def_G_deterministic}
\end{align}
For \(\epsilon = 0\), the capacity is
\begin{align}
C_{\text{global-GC}}(0) = \max_{p_0 \geq 1-\beta, \; p_k \geq \alpha, \; p_0 + p_k \leq 1} G(p_0, p_k).
\end{align}
If \(\frac{\alpha}{k+1} \geq \frac{\beta - \alpha}{m_k}\) and \(\frac{1-\beta}{k+1} \geq \frac{\beta - \alpha}{m_k}\), the maximum is attained at the boundary, i.e., \(C_{\text{global-GC}}(0) = G(1-\beta, \alpha)\).
\end{corollary}

\begin{proof}
Since every symbol in \(\mathcal{B}_k\) always and every symbol in \(\mathcal{B}_1, \dots, \mathcal{B}_{k-1}\) possibly manifests as G or C, the smallest and largest values of \(Y_\mathbf{t}\) with nonzero probability are \(t_k\) and \(n - t_0\). Hence, \(\mathcal{P}_0 = \{\mathbf{t} : t_k \geq n\alpha, \; t_0 \geq n(1-\beta)\}\), in accordance with Definition~\ref{def:det_gc}. The constraint does not distinguish between the bins \(\mathcal{B}_1, \dots, \mathcal{B}_{k-1}\), so for a fixed total mass \(1 - p_0 - p_k\) on these bins, \(F\) is maximized by the uniform distribution over all \(m_k\) mixed symbols. Together with \(|\mathcal{B}_0| = |\mathcal{B}_k| = k+1\) and the continuity of \(F\), Theorem~\ref{prop:global_capacity} yields~(\ref{eq:def_G_deterministic}). The objective is the entropy of a distribution over \(\mathcal{A}_k\) that is uniform within the pure A/T, pure G/C, and mixed symbols, and is therefore concave. At \(p_0 = 1-\beta\), \(p_k = \alpha\), the stated conditions require the per-symbol probability of both constrained groups to be at least that of the mixed symbols, which is exactly the Karush--Kuhn--Tucker condition for both constraints to be active.
\end{proof}

For example, for \(k = 4\), \(\alpha = 0.4\) and \(\beta = 0.6\), the conditions hold and \(C_{\text{global-GC}}(0) \approx 4.31\) bits, i.e., \(84\%\) of \(\log_2 |\mathcal{A}_4|\), whereas \(C_{\text{global-GC}}(\epsilon) = \log_2 |\mathcal{A}_4|\) for every \(\epsilon > 0\) by Remark~\ref{rem:eps_asymptotic}.

\begin{corollary}[Deterministic Worst-Case Baseline]
Under the deterministic baseline (\(\epsilon = 0\)), intermediate probabilities are irrelevant: a window is admissible if and only if it contains at least \(w\alpha\) pure G/C symbols and at most \(w\beta\) symbols that are not pure A/T. It is therefore sufficient to partition the composite alphabet into three bins: pure A/T \(\mathcal{B}_{\text{AT}} = \mathcal{B}_0\), pure G/C \(\mathcal{B}_{\text{GC}} = \mathcal{B}_k\), and mixed \(\mathcal{B}_{\text{mix}} = \bigcup_{i=1}^{k-1} \mathcal{B}_i\) with \(|\mathcal{B}_{\text{mix}}| = m_k\) as in Corollary~\ref{cor:det_capacity}. The state space reduces to \(3^{w-1}\) states, and \(C_{\mathrm{w-GC}}(0) = \log_2 \lambda(T_0)\), where \(T_0\) is the weighted transition matrix over this 3-bin partition.
\end{corollary}

Note that \(T_\epsilon = T_0\) whenever \(\epsilon\) is smaller than the smallest nonzero violation probability of any length-\(w\) window, in which case the probabilistic and deterministic sliding-window capacities coincide. For the parameters of Remark~\ref{rem:eps_window} and \(w=5\), this minimum is \(1/16\), attained, e.g., by a window with one symbol from \(\mathcal{B}_0\), two from \(\mathcal{B}_1\), and two from \(\mathcal{B}_4\). Hence, \(\epsilon = 0.05\) yields the same capacity as \(\epsilon = 0\).

\section{Examples}\label{app:examples}

Intuitively, as the window shifts by one position, the oldest symbol leaves the window and a new symbol enters. The state therefore records the overlapping portion between consecutive windows, while each transition corresponds to appending a new probability bin. A transition is allowed whenever the resulting window satisfies the local \(\epsilon\)-probabilistic GC constraint.

\begin{example}
Figure~\ref{fig:fsm_transition} illustrates the state transitions for a window of length \(w=3\) using an alphabet partitioned into three bins. Because \(w=3\), each state is defined by a length-\(2\) sequence of bins (representing the \(w-1\) overlap). If the system is currently in state \((b_1, b_1)\), appending a new bin creates a temporary length-\(3\) window. The transition is permitted if and only if this resulting window satisfies the \(\epsilon\)-probabilistic GC bounds. Assume that appending \(b_1\) or \(b_3\) violates the constraint, while appending \(b_2\) forms a valid window. Consequently, the oldest bin is discarded, and the system successfully transitions to the new state \((b_1, b_2)\).
\end{example}

\begin{figure}
    \centering
    \includegraphics[width=\columnwidth]{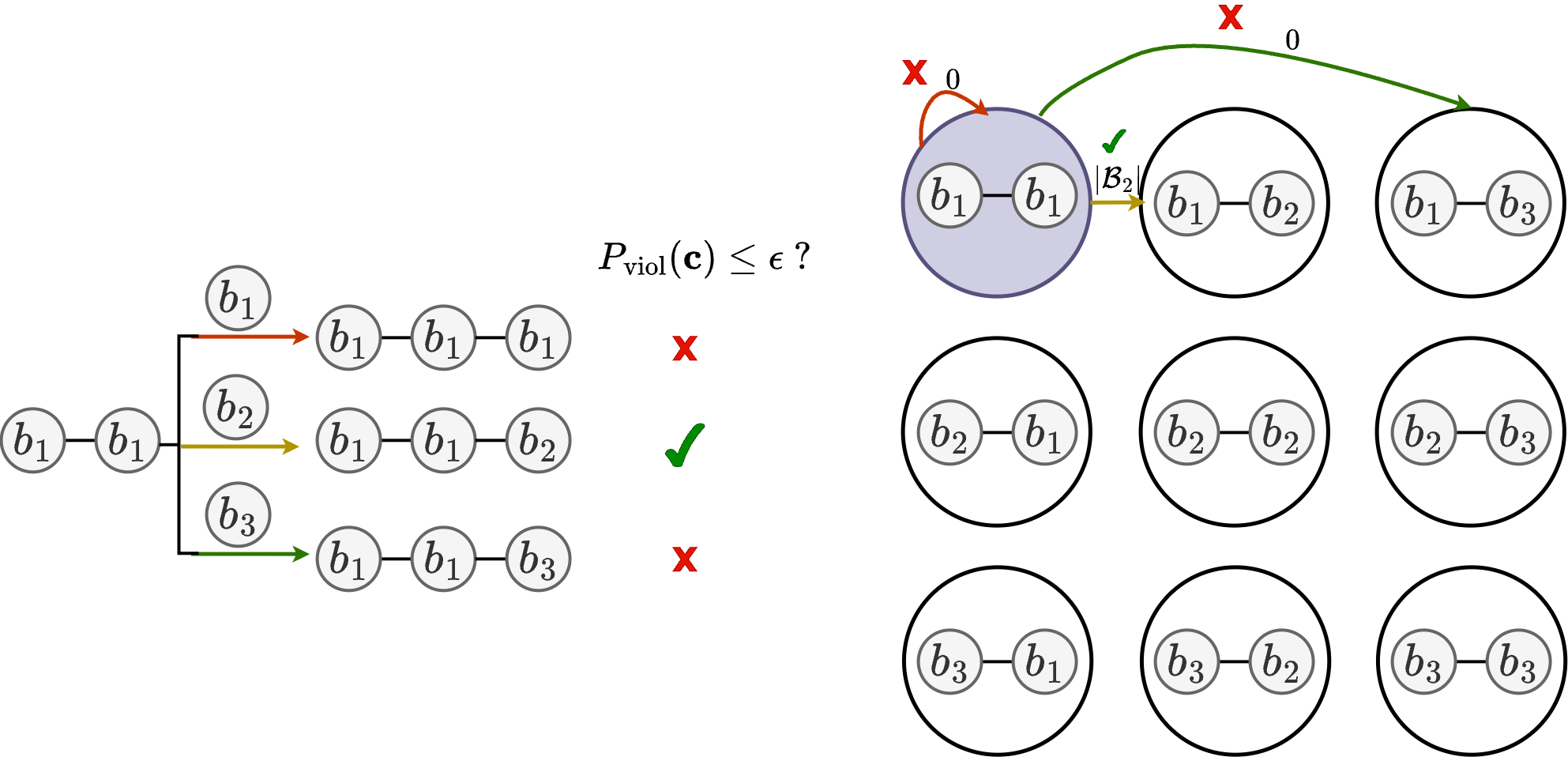}
    \caption{State transition logic for a sliding-window constraint with \(w=3\) and three probability bins (\(b_1, b_2, b_3\)). From the current state \((b_1, b_1)\), appending a new bin creates a length-\(3\) window. Only the addition of \(b_2\) satisfies the GC constraint, permitting a valid transition to the next overlapping state \((b_1, b_2)\).}
    \label{fig:fsm_transition}
\end{figure}

Intuitively, this process can be viewed as a dynamic dictionary lookup. Imagine all probabilistically valid composite sequences sorted lexicographically and laid out uniformly across a continuous line from 0 to 1. A binary message simply acts as a specific coordinate on this line. Instead of generating and storing this massive dictionary beforehand, the encoder constructs the target sequence symbol by symbol, recursively narrowing its search window based on the proportion of valid sequences that remain.
Figure~\ref{fig:arithmetic_tree} visualizes the interval partitioning logic described in Section~\ref{sec:coding_schemes}.

\begin{figure}[htpb]
    \centering
    \includegraphics[width=\columnwidth]{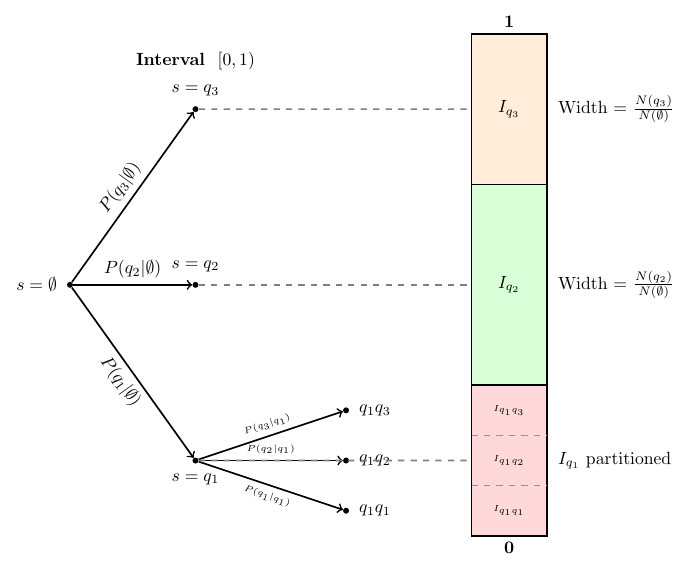} 
    \caption{Conceptual tree and interval partitioning of the Multi-Type Arithmetic Encoder. Extending the prefix by one additional step (\(s = a_1 a_j\)) recursively subdivides the corresponding sub-interval \(I_{a_1}\) proportional to the multi-type continuation counts \(N(a_1 a_j)\).}
    \label{fig:arithmetic_tree}
\end{figure}

\begin{example}[Multi-Type Encoding Trace]
To illustrate the step-by-step mechanics of Algorithm~\ref{alg:encoder}, consider a simplified scenario where not all possible compositions are present and the \(\epsilon\)-probabilistic constraint allows exactly two valid composition types for a strand of length \(n=3\): 
\(\mathcal{P}_\epsilon = \{\mathbf{t}_1, \mathbf{t}_2\}\), where \(\mathbf{t}_1 = (1,1,1)\) and \(\mathbf{t}_2 = (2,1,0)\).

As shown in Figure \ref{fig:worked_example}, the composite alphabet is partitioned into three probability bins (\( \mathcal{B}_0, \mathcal{B}_1, \mathcal{B}_2\)) 
containing 2, 2, and 3 symbols, respectively. Note that this is a toy example: the bin sizes are chosen for illustration and do not correspond to a full composite alphabet \(\mathcal{A}_k\), for which \(|\mathcal{B}_i| = (i+1)(k-i+1)\) would hold.

\begin{figure}[htpb]
    \centering
    \includegraphics[width=1\columnwidth]{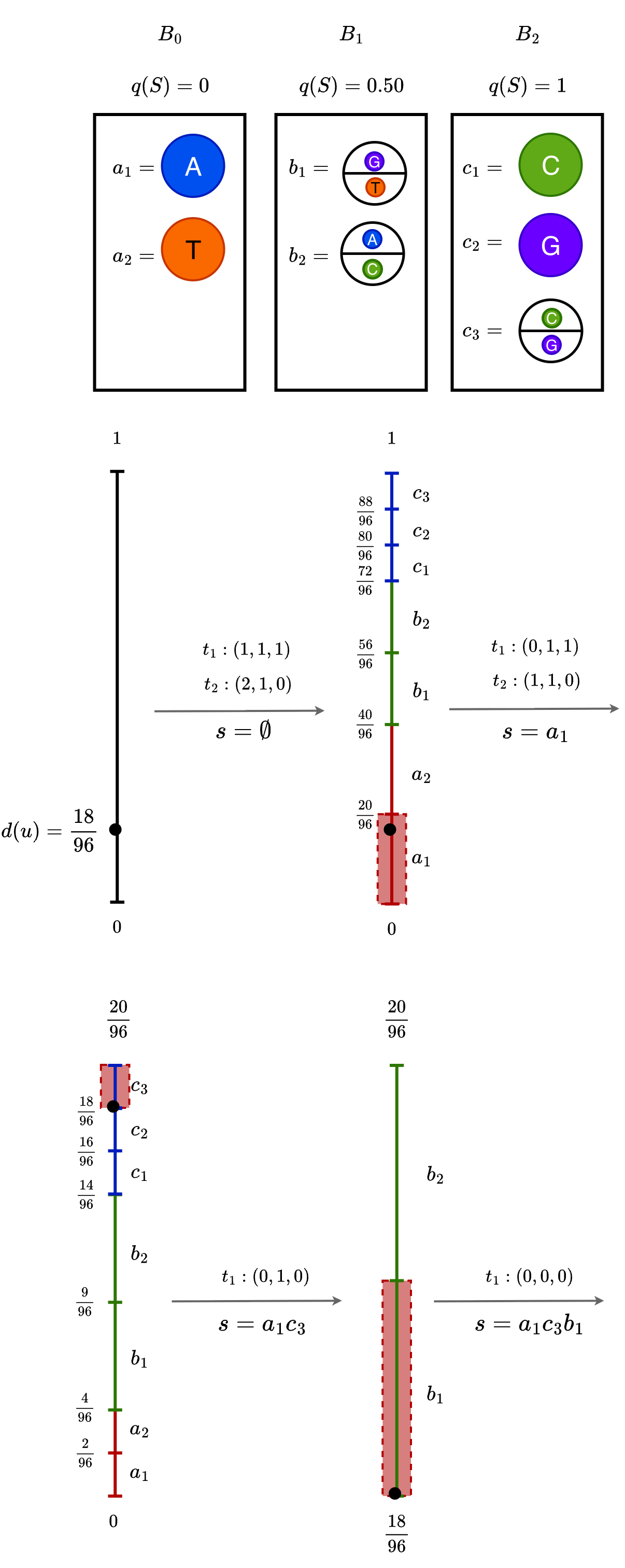}
    \caption{A numerical trace of the multi-type arithmetic encoder for a target fraction \(d = 0.1875\).}
    \label{fig:worked_example}
\end{figure}

By Proposition~\ref{prop:class_cardinality}, we evaluate the cardinality of each allowed type-class:
\begin{itemize}
    \item \(|\mathcal{T}_{\mathbf{t}_1}| = \frac{3!}{1!1!1!} (2^1 \cdot 2^1 \cdot 3^1) = 72\) sequences.
    \item \(|\mathcal{T}_{\mathbf{t}_2}| = \frac{3!}{2!1!0!} (2^2 \cdot 2^1 \cdot 3^0) = 24\) sequences.
\end{itemize}
The total valid codebook size is \(N(\emptyset) = 72 + 24 = 96\) valid sequences. Thus, the encoder can map \(m = \lfloor \log_2(96) \rfloor = 6\) bits. 

Suppose we wish to encode the binary message \(\mathbf{b} = 001100\). The target fraction is computed as \(d = \operatorname{NBC}(001100) / 2^6 = 12/64 = 18/96 = 0.1875\).

The encoder initializes the interval \([0,1)\) and evaluates the candidate symbols for the first position. For candidate \(a_1 \in \mathcal{B}_0\), we must compute the unified prefix-completion count \(N(a_1)\) by summing the continuations from both types via Proposition~\ref{prop:prefix_count}:
\begin{itemize}
    \item For \(\mathbf{t}_1\), the remaining requirement is \((0,1,1)\). The continuations are \(N_{\mathbf{t}_1}(a_1) = \frac{2!}{0!1!1!} (2^0 \cdot 2^1 \cdot 3^1) = 12\).
    \item For \(\mathbf{t}_2\), the remaining requirement is \((1,1,0)\). The continuations are \(N_{\mathbf{t}_2}(a_1) = \frac{2!}{1!1!0!} (2^1 \cdot 2^1 \cdot 3^0) = 8\).
\end{itemize}
Thus, the total valid continuations are \(N(a_1) = 12 + 8 = 20\). 

The corresponding interval width for \(a_1\) is proportional to this sum: \(P(a_1|\emptyset) = 20/96\). Because our target fraction \(d = 18/96\) falls within the sub-interval \([0, 20/96)\), the encoder successfully selects \(a_1\). This process recurses, dynamically pruning types when their requirements drop below zero, until the final length-$3$ codeword is isolated.
\end{example}

\end{document}